\documentclass[conference]{IEEEtran}
\usepackage[T1]{fontenc}
\usepackage[utf8]{inputenc}
\usepackage{blindtext, graphicx}
\usepackage{tabularx}
\usepackage{bbm}
\usepackage{pgf,tikz,pgfplots}
\usetikzlibrary[patterns]
\usepackage[ruled,vlined]{algorithm2e}
\usepackage{hyperref}
\usepackage{mathrsfs}
\usetikzlibrary{arrows}
\hypersetup{
	colorlinks,
	linkcolor={blue!100!black},
	citecolor={blue!100!black},
	urlcolor={blue!80!black}
}

\usepackage{amsmath,amssymb,amsfonts,amsthm,graphicx,bm,bbm,dsfont}

\newtheorem{lemma}{Lemma}

\usepackage{soul}

\DeclareSymbolFont{bbold}{U}{bbold}{m}{n}
\DeclareSymbolFontAlphabet{\mathbbold}{bbold}

\usepackage{color}

\newcommand{\bolda}{\mathbf{a}}
\newcommand{\boldb}{\mathbf{b}}
\newcommand{\boldc}{\mathbf{c}}

\newcommand{\boldp}{\mathbf{p}}
\newcommand{\boldq}{\mathbf{q}}

\newcommand{\boldv}{\mathbf{v}}

\makeatletter
\def\namedlabel#1#2{\begingroup
	\def\@currentlabel{#2}%
	\label{#1}\endgroup
}
\makeatother
\begin{document}
	\title{On Constant-Weight Binary $B_2$-Sequences} 
 \author{
   \IEEEauthorblockN{
   \textbf{Jin Sima}, \textbf{Yun-Han Li}, \textbf{Ilan Shomorony}  and \textbf{Olgica Milenkovic}}
   \IEEEauthorblockA{Department of Electrical and Computer Engineering, University of Illinois Urbana-Champaign, USA \\\texttt{\{jsima,yunhanl2,ilans,milenkov\}@illinois.edu}
   }
 }
	\maketitle
	
\begin{abstract}
Motivated by applications in polymer-based data storage we introduced the new problem of characterizing the code rate and designing constant-weight binary $B_2$-sequences. Binary $B_2$-sequences are collections of binary strings of length $n$ with the property that the real-valued sums of all distinct pairs of strings are distinct. In addition to this defining property, constant-weight binary $B_2$-sequences also satisfy the constraint that each string has a fixed, relatively small weight $\omega$ that scales linearly with $n$. The constant-weight constraint ensures low-cost synthesis and uniform processing of the data readout via tandem mass spectrometers. Our main results include upper bounds on the size of the codes formulated as entropy-optimization problems and constructive lower bounds based on Sidon sequences.
\end{abstract}
	
\thispagestyle{empty}

\section{Introduction}\label{section:introduction}	
Binary $B_2$-sequences  were introduced by Lindstr\"om in~\cite{lindstrom1969determination} and were subsequently studied in a number of follow-up works~\cite{lindstrom1972b2,cohen2001binary,della2022note}. Binary $B_2$-sequences represent a set (codebook) of binary vectors of some fixed length such that the entry-wise real-valued sums of all pairs of codevectors from the set are distinct; hence, given the sum one can uniquely determine the vectors that were summed up. Since their introduction, these sequences have found many applications such as for search algorithms~\cite{d2014lectures,bouvel2005combinatorial}, multiple access system design~\cite{gritsenko2017signature,polyanskiy2018information}, and data fingerprinting~\cite{boneh1998collusion,cohen2001binary}. 

Some more recent applications of binary $B_2$-sequences include \emph{polymer-based data storage.} Nonvolatile storage systems based on DNA and other native macromolecules and synthetic polymers hold the promise of ultrahigh storage densities, long-term readout compatibility and exceptional durability~\cite{goldman2013towards,Church2012Next,grass2015robust,yazdi2015rewritable,yazdi2017portable,tabatabaei2020dna,pan2021rewritable,tabatabaei2022expanding}. Synthetic polymers are binary molecular storage media that represent $0$s and $1$s using polymers of significantly different masses~\cite{laure2016coding}. A user-defined binary string is created by stitching together the polymer symbols in the required order and it is read by measuring the masses of prefixes and suffixes (or all substrings) of the polymer strings~\cite{acharya2015string,pattabiraman2019reconstruction}. To ensure unique reconstruction of mixture of polymer strings based on their prefix and suffix compositions only, one needs to follow a more involved process, described in~\cite{gabrys2021reconstructing,gabrys2022reconstruction}. There, binary $B_h$-sequences (codes)~\cite{lindstrom1969determination,cohen2001binary,della2022note} are used to ensure that the sums of masses of prefixes of the same length uniquely determine the strings themselves. Since in practice the polymer used to represent $1$s has a significantly higher mass than the polymer used to represent $0$s, the mass discrepancy can lead to high fragmentation loss and significantly increased chemical synthesis cost, using binary $B_2$-sequences of relatively small weight is desirable. This motivates introducing the problem of \emph{constant-weight $B_2$-sequence} design. 
 
The main results of our work include information-theoretic, prefix-suffix splitting upper bounds on the size of constant-weight binary $B_2$-sequences for which the weight $\omega$ scales linearly with $n$. Unlike its unconstrained counterpart, the strongest upper bound is given in terms of an optimization problem that has to be solved numerically. In addition, we also provide constructive lower bounds based on Sidon sequences.

The paper is organized as follows. Section~\ref{section:preliminary} presents the notation, relevant concepts and a generalization of the approach from~\cite{lindstrom1969determination} to the case of constant-weight binary strings. Section~\ref{section:boundsdalai} contains our main result, the sharpest known upper bound on the size of binary constant-weight $B_2$-sequences. Constructive lower bounds are presented in Section~\ref{sec:lowerbound}.
 
\section{Preliminaries and Entropy Bounds}\label{section:preliminary}

We denote sets by calligraphic upper-case letters and vectors by boldface lower-case letters. Cardinalities of sets are denoted by upper-case letters. We also use $[n]$ to denote the set $\{{1,2,\ldots,n\}}$. All logarithms, unless stated otherwise, are taken base-$2$.

A set $\mathcal{A}_n\subset\{0,1\}^n$ of binary vectors is called a $B_2$-sequence set if real-valued sums of all distinct pairs of strings $\boldc_1+\boldc_2$, $\boldc_1,\boldc_2\in\mathcal{A}_n,$ are distinct. A $B_2$-sequence set $\mathcal{A}^\omega_n$ is said to have constant weight $\omega$ if every vector $\boldc=(c_1,\ldots,c_n)\in \mathcal{A}^\omega_n$ has Hamming weight $|\boldc|\triangleq\sum^n_{i=1}c_i=\omega$. 
 
 Let $A^\omega_n$ be the size of the largest constant-weight $B_2$-sequence set $\mathcal{A}^\omega_n$ of weight $\omega$. We are interested in the asymptotic behavior of $A^\omega_n$, for constant $\bar\omega=\frac{\omega}{n}$, or more precisely, in the asymptotic code rate $R_{\bar\omega}=\limsup_{n\rightarrow \infty}\frac{\log A^{\bar\omega n}_n}{n}$. The results established in~\cite{lindstrom1969determination,cohen2001binary} imply that for unrestricted binary $B_2$-sequence sets, the asymptotic code rate satisfies $\le 0.5753$, which also establishes $R_{\bar\omega} \le 0.5753$ for any $\bar\omega\in[0,1]$. We seek to improve this upper bound on $R_{\bar\omega}$ for $\bar\omega\in [0,1/2)$.

A simple asymptotic upper bound on $R_{\bar\omega}$ can be derived using information-theoretic arguments (attributed to Katona for the case of unrestricted sequences), by assuming a \emph{uniform probability distribution} on the set of all possible \textit{ordered} pairs $(\boldsymbol{C},\boldsymbol{C}')$ of codevectors and invoking the fact that each pair results in a unique sum. Note that $C$ and $C'$ are chosen independently but are allowed to be equal. To obtain an upper bound on $R_{\bar\omega}$, we assume that the size of the constant-weight $B_2$ code of length $n$ equals $A_n^{\omega}$. Then, the entropy of all ordered pairs of codevectors equals $2\, \log(A_n^{\omega})$. 
Let $G=\mathbbm{1}( \boldsymbol{C}>\boldsymbol{C}')$ be an indicator random variable for the event that $\boldsymbol{c}$ is lexicographically ranked higher than $\boldsymbol{c}'$.  Since the sums of all unordered pairs are all distinct, 
there is a bijection between the conditional probability space of the pair $(\boldsymbol{C},\boldsymbol{C}')$ given $G$ and the probability space of a sequence of random variables, $X_i=C_i+C'_i$, $i\in[n]$, representing the coordinates of the sum vector. 
Consequently,
\begin{align}\label{eq:subadditivity}
2\log(A^{\omega}_n) &= H(\boldsymbol{C},\boldsymbol{C}'|G)+H(G) \\ \notag
&=H(X_1,X_2,\ldots,X_n|G)+H(G)\\\notag
&\leq H(X_1)+H(X_2)+\ldots+H(X_n)+1. \notag
\end{align}
Assume that the probability of observing a $1$ at the $i$th coordinate (i.e., $C_i$) of the $B_2$-sequence code equals $p_i$, $i=1,\ldots,n$. Given that the weight of the binary vectors is $\omega$, we have $\frac{1}{n}(p_1+p_2+\ldots+p_n)=\frac{\omega}{n}=\bar\omega$. The entropy of the $i$th coordinate of all possible $2$-sums equals $H(X_i)$ and is equal to the entropy of a Binomial$(2,p_i)$ distribution, i.e., the distribution $\textbf{p}^i$ with probabilities of $0$,
$1$ and $2$ equal to
$$p^i_{0}=\left(1-p_i\right)^2,\; p^i_{1}=2\, p_i\,\left(1-p_i\right),\; p^i_{2}= p_i^2, \; i\in[n]$$
respectively. Note that $H(X_i)=H(\boldp_i)$ equals
\begin{align}
&-p^2_i\log p^2_i-(1-p_i)^2\log(1-p_i)^2-2p_i(1-p_i)\log p_i\nonumber\\
&- 2p_i(1-p_i)\log (1-p_i)-2p_i(1-p_i)\nonumber\\
=&-2p_i\log p_i-2(1-p_i)\log(1-p_i)-2p_i+2p^2_i=H_{bin}(p_i). \nonumber
\end{align}
Since 
$\frac{d^2H_{bin}(p_i)}{d p^2_i}=-\frac{2}{\ln (2)}(\frac{1}{p_i}+\frac{1}{1-p_i})+4<0$, the function $H_{bin}(p_i)$ is concave in $p_i$. 
Hence, we have 
\begin{align}\label{eq:sumofentropy} &H(X_1)+H(X_2)+\ldots+H(X_n)\nonumber\\
& =H_{bin}(p_1)+H_{bin}(p_2)+\ldots+H_{bin}(p_n)\nonumber\\
& \le nH_{bin}\left(\frac{\sum^n_{i=1}p_i}{n}\right)=nH_{bin}\left(\bar\omega\right).
\end{align}
Note that $H_{bin}(\bar\omega)$ stands for the entropy of a Binomial$(2,\frac{\omega}{n})$ distribution, i.e., the distribution
$$p_0=\left(1-\bar\omega\right)^2,\; p_1=2\, \bar\omega\,\left(1-\bar\omega\right),\; p_2= \bar\omega^2.$$
By \eqref{eq:subadditivity} and \eqref{eq:sumofentropy} it follows that $2\log(A_n^{\omega}) \leq n  H_{bin}\left(\bar\omega\right)+1$ and
\begin{align}\label{eq:entropyboundsum}
R_{\bar\omega} \leq \frac{1}{2}\, H_{bin}\left(\bar\omega\right).   
\end{align}
While \eqref{eq:entropyboundsum} provides a good starting upper bound for $R_{\bar\omega}$, we note that an alternative entropy bound can be obtained by considering the coordinates of the difference of codevector, rather than the sum of codevectors. Specifically, we assume a uniform distribution on the set of ordered pairs of codevectors $(\boldsymbol{C},\boldsymbol{C}')$ in a constant-weight $B_2$ code, so that the entropy of the ordered pair is as before given by $2\log (A_n^\omega)$. Let $Y_i=C_i-C'_i$, $i\in[n]$, be the sequence of random variables 
representing the values of coordinates of the difference of vector in the ordered pair. Then, either one of the following holds: (1) $Y_i=0$ for $i\in[n]$, whenever the two vectors in the ordered pair are equal; (2) there is a one to one mapping between $Y_1,Y_2,\ldots,Y_n$ and the ordered pair if the two vectors in the pair are not equal. This follows because if any two different pairs $\boldc_1,\boldc_2\in A^\omega_n$ and $\boldc_1,\boldc_2\in A^\omega_n$ have the same difference $\boldc_1-\boldc_2=\boldc_3-\boldc_4$, then the two pairs $\boldc_1,\boldc_4$ and $\boldc_2,\boldc_3$ 
have the same sum $\boldc_1+\boldc_4=\boldc_2+\boldc_3$, which violates the $B_2$ constraint. Next, let $E$ be the event $\{(Y_1,\ldots,Y_n)=0^n\}$. Then, we can upper-bound $2\log (A_n^\omega)$ by 
\begin{align*}
   &H(C,C'|Y_1,Y_2,\ldots,Y_n)+H(Y_1,Y_2,\ldots,Y_n)\\
   =&\left(1-\frac{1}{A_n^\omega}\right)H(C,C'|Y_1,Y_2,\ldots,Y_n,E^c)\\
   &+\frac{1}{A_n^\omega}H(C,C'|Y_1,Y_2,\ldots,Y_n,E)
   +H(Y_1,Y_2,\ldots,Y_n)\\
   =&\frac{\log A_n^\omega}{A_n^\omega}+H(Y_1,Y_2,\ldots,Y_n)\\
   \le &H(Y_1)+H(Y_2)+\ldots+H(Y_n)+\frac{\log A_n^\omega}{A_n^\omega},
\end{align*}
where $(1-\frac{1}{A_n^\omega})$ equals the probability that the two codevectors in the pair are not equal. Assume that $P\{{C_i=1\}}=q_i$. Then, as before, we have $\frac{1}{n}(q_1+q_2+\ldots +q_n)=\bar\omega$. The distribution $\boldq^i$ of $Y_i$ is given by
$$q^i_0=q^2_i+(1-q_i)^2,\; q^i_1=q_i(1-q_i), \;q^i_{-1}=q_i(1-q_i),\; i\in[n].$$
Since the entropy function $H(\boldq^i)$ is concave in $\boldq^i$, 
\begin{align}
    &H(Y_1)+H(Y_2)+\ldots+H(Y_n) \nonumber \\
    =&H(\boldq^1)+H(\boldq^2)+\ldots+H(\boldq^n)\nonumber\\
     \le &nH(\frac{\boldq^1+\boldq^2+\ldots+\boldq^n}{n})\nonumber\\
    =&n\Big(-t\log t-(1-t)\log (\frac{1-t}{2})\Big), \nonumber
\end{align}
where $t = \frac{\sum^n_{i=1}[q^2_i+(1-q_i)^2]}{n}$. Note that $t\ge \bar\omega^2+(1-\bar\omega)^2\ge \frac{1}{2}$ by Jensen's inequality, and that for $t\ge \frac{1}{2}$ the function $-t\log t-(1-t)\log (\frac{1-t}{2})$ decreases as $t$ increases. Hence, by the previous inequality we have
\begin{align*}
    H(Y_1)+H(Y_2)+\ldots+H(Y_n)\le nH(\boldq),
\end{align*}
where the distribution $\boldq$ is given by
$$q_0=\bar\omega^2+(1-\bar\omega)^2,\;q_1=\bar\omega(1-\bar\omega),\;q_{-1}=\bar\omega(1-\bar\omega).$$
As a result, $2\log(A_n^{\omega}) \leq n  H(\boldq)+\frac{\log A_n^{\omega}}{A_n^{\omega}},$
and
\begin{align}\label{eq:entropybounddiff}
R_{\bar\omega} \leq \frac{1}{2} H(\boldq).    
\end{align}
It can easily be shown that the upper bound~\eqref{eq:entropyboundsum} is strictly better than~\eqref{eq:entropybounddiff}. However, as we will see later, the bound \eqref{eq:entropybounddiff} is more useful for deriving upper bounds than those obtained from direct information-theoretic arguments. 

In Table \ref{tab:results}, we provide numerical values for the information-theoretic bounds derived in this section, along with those of the improved upper bounds to be described in the subsequent exposition, for different values of weights $\omega$ that scale linearly with $n$. As may be seen, as $\omega$ decreases, all upper bounds become closer in value. Additional results included in the table are lower bounds, discussed in more detail in the last section.

\begin{table}[h!]
\caption{Upper and Lower Bounds on the size of binary, constant-weight $B_2$-sequences}
\label{tab:results}
\centering
\begin{tabular}{ |p{2.07cm}||p{0.46cm}|p{0.46cm}|p{0.46cm}|p{0.46cm}|p{0.46cm}|p{0.46cm}|p{0.46cm}| }
 \hline
 The value of $\bar\omega$& 0.5&0.4&0.345&0.2&0.1&0.05&0.02\\
 \hline
 Entropy bound \eqref{eq:entropyboundsum}&0.75&0.731&0.704&0.562&0.379&0.239&0.122\\
 \hline
 Entropy bound \eqref{eq:entropybounddiff}&0.75&0.739&0.723&0.612&0.43&0.274&0.139\\
 \hline
 Upper bound \eqref{eq:upperboundlindstrom} &0.6&0.6&0.594&0.515&0.365&0.235&0.121\\
  \hline
  Upper bound \eqref{eq:upperbounddalai} &0.6&0.59&0.575&0.487&0.349&0.228&0.12\\
 \hline
    Lower bound 
    &0.25&0.259&0.263&0.232&0.166&0.108&0.056\\
 \hline 
\end{tabular}
\vspace{-0.1in}
\end{table}

\vspace{-0.25in}
\subsection{Improved Upper Bounds}\label{section:boundsdalai}
An improved upper bound on $R_{\bar\omega}$ for our constrained $B_2$ codebook design can be obtained by adapting and generalizing a recent approach from~\cite{della2022note} for unconstrained $B_2$ vectors. The underlying proof combines entropy bounds with a prefix-suffix decomposition approach first reported in~\cite{lindstrom1972b2}. For completeness, we first describe how to extend the prefix-suffix decomposition approach for constant-weight $B_2$ codebooks. Afterwards, we improve the two bounds -- the information-theoretic and prefix-suffix bound -- by combining them~\cite{della2022note}. The main differences between the approaches designed for general codebooks~\cite{della2022note,lindstrom1972b2} and our approach is that we use more elaborate entropy bounds, group the codevectors $\boldc\in\mathcal{A}^\omega_n$ based on the weight of their prefixes and invoke specialized counting techniques. Importantly, the approach in~\cite{della2022note} does not improve the result that can be obtained purely through the use of prefix-suffix decompositions~\cite{lindstrom1972b2}, while our scheme improves \emph{both the entropy and prefix-suffix approach} for constant-weight codebooks. The proof is deferred to Appendix~\ref{sec:proofoflindstrom}. 

\begin{lemma}\label{lemma:weightlindstrom} 
Let every codevector be split as $\boldc=\bolda\boldb\in \mathcal{A}^\omega_n$, where $\bolda\in\{0,1\}^e$ and $\boldb\in\{0,1\}^{n-e}$. Let $e=\bar e n$, where $\bar e$ is a constant in $[0,1]$. 
Then, $R_{\bar\omega}$ can be upper bounded by
\begin{align}\label{eq:upperboundlindstrom}
&\min_{\bar e\in[0,1]}\max_{\bar\omega'\in[\max\{0,\bar\omega-1+\bar e\},\min\{\bar e,\bar\omega\}]}\max\Bigg\{H(\frac{\bar\omega'}{\bar e})\cdot \bar e, \nonumber\\
    &\frac{1}{2}\Big[H(\frac{\bar\omega'}{\bar e})\cdot \bar e
  +H(\frac{2\bar\omega''}{1-\bar e})\cdot(1-\bar e)+2\bar\omega''\Big]\Bigg\},
\end{align}
where $\bar\omega''=\min\{\bar\omega-\bar\omega',\frac{1-\bar e}{4}\}$. By optimizing over $\bar e$, numerical values for the bound can be found for different $\omega'$s. A sampling of the results is shown in Table~\ref{tab:results}.
 \end{lemma}
Based on the above result and its proof, we describe next our main result, constituting a sharper asymptotic upper bound on constant-weight binary $B_2$-sequences. Let
\begin{align}
\mathcal{B}^{\omega'}_n=\{\boldc:\boldc=\bolda\boldb\in\mathcal{A}^\omega_n,\;\bolda\in\{0,1\}^e,\;|\bolda|=\omega'\},
\end{align}
where $e=\bar e n$ and $\omega'=\bar\omega' n$ are constants s.t. $\bar e,\bar\omega'\in[0,1]$,
be the set of codevectors in $\mathcal{A}^\omega_n$ whose prefixes of length $e$ have weight $\omega'\in[\max\{0,\omega-n+e\},\min\{e,\omega\}]$. For notational convenience, we also use $f=n-e$ to denote the length of the suffixes. 
Note that 
\begin{align}\label{eq:aomega}
    A^\omega_n=\sum_{\omega'\in[\max\{0,\omega-f\},\min\{e,\omega\}]}|\mathcal{B}^{\omega'}_n|.
\end{align}
Hence, we need to establish an upper bound on $|\mathcal{B}^{\omega'}_n|$ for any $n$, $e$ and $\omega'\in [\max\{0,\omega-f\},\min\{e,\omega\}]$.
\begin{lemma}\label{lemma:weightdalai}
For any $\omega'\in[\max\{0,\omega-f\},\min\{e,\omega\}]$, we have
\begin{align}\label{eq:weightdalai1}
 \log |\mathcal{B}^{\omega'}_n|\le &\max\Bigg\{e\,H\left(\frac{\omega'}{e}\right)+\log n, \nonumber\\
  &\frac{1}{2}\Big[e\,H\left(\frac{\omega'}{e}\right)+fH(p_0,p_1,p_{-1})\Big]+1\Bigg\}, 
\end{align}
where $p_0=\frac{(\omega'')^2+(f-\omega'')^2}{f^2}$,  $p_1=p_{-1}=\frac{1-p_0}{2}$, and 
$\omega''=\omega-\omega'$. The function $H(x)=-x\log x-(1-x)\log (1-x)$ stands for the binary Shannon entropy function, while the function $H(p_0,p_1,p_{-1})=-p_0\log p_0-p_1\log p_1-p_{-1}\log p_{-1}$ stands for the entropy of a ternary random variable with the distribution $(p_0,p_1,p_{-1})$ described above. 
\end{lemma}
\begin{proof}
    Fix $\omega'$ and set $B^{\omega'}_n=|\mathcal{B}^{\omega'}_n|$. 
    Let $\{\bolda_1,\ldots,\bolda_{r}\}=\{\bolda:\bolda\boldb\in\mathcal{B}^{\omega'}_n \text{ for some $\boldb$}\}$ be the set of all possible prefixes of the codevectors in $\mathcal{B}^{\omega'}_n$. Note that each $\bolda_i$ has weight $\omega'$ and that there are at most $r$ different such vectors, where $r\le \binom{e}{\omega'}$. Let $\mathcal{S}_i=\{\boldb:\bolda_i\boldb\in\mathcal{B}^{\omega'}_n\}$, $i\in[r],$ be the (possibly empty) set of suffixes of codevectors in $\mathcal{B}^{\omega'}_n$ that have prefix $\bolda_i$. Then, 
    \begin{align}
    B^{\omega'}_n= \sum^r_{i=1}|\mathcal{S}_i|.
    \end{align}    
    Now consider the set of all pairs of suffixes that belong to the same group $\mathcal{S}_i$ for some $i\in[r]$, denoted by
    \begin{align}
    \mathcal{D}=\cup^r_{i=1}\{(\boldb_1,\boldb_2):\boldb_1,\boldb_2\in\mathcal{S}_i\}.
    \end{align}
    Note that $\boldb_1$ and $\boldb_2$ are allowed to be the same and that $(\boldb_1,\boldb_2)$ and $(\boldb_2,\boldb_1)$ are considered two different pairs, provided $\boldb_1,\boldb_2\in\mathcal{S}_i$, $i\in [r]$, are distinct. Hence, $\mathcal{D}$ is a multiset. 
     Then
     \begin{align} \label{eq:cauchy}|\mathcal{D}|=\sum^r_{i=1}|\mathcal{S}_i|^2\ge\frac{(\sum^r_i|\mathcal{S}_i|)^2}{r}=\frac{(B^{\omega'}_n)^2}{r}, 
     \end{align}
     where the bound follows from Cauchy-Schwarz's inequality. 
     Furthermore, consider the differences between all pairs in $\mathcal{D}$,     
     \begin{align}
        \mathcal{Z}=\{\boldb_1-\boldb_2:(\boldb_1,\boldb_2)\in \mathcal{D}\},
    \end{align}
    where $\mathcal{Z}$ is a multiset. The multiplicity of $0^f$ (the all $0$ vector of length $f$) in $\mathcal{Z}$ is exactly $B^{\omega'}_n$. In addition, the multiplicity of any nonzero element in $\mathcal{Z}$ is exactly one. To see this, suppose on the contrary that there exist different pairs of unequal elements $(\boldb_1,\boldb_2),(\boldb_3,\boldb_4)\in\mathcal{D}$ satisfying $\boldb_1-\boldb_2=\boldb_3-\boldb_4$. By definition of $\mathcal{D}$, we have that $\boldb_1,\boldb_2\in\mathcal{S}_i$ for some $i\in[r]$ and $\boldb_3,\boldb_4\in\mathcal{S}_j$ for some $j\in[1,r]$. This implies that $\bolda_i\boldb_1,\bolda_i\boldb_2,\bolda_j\boldb_3,\bolda_j\boldb_4$ are codevectors in $\mathcal{B}^{\omega'}_n$. Then, 
    \begin{align}
        \bolda_i\boldb_1-\bolda_i\boldb_2=\bolda_j\boldb_3-\bolda_j\boldb_4,
    \end{align}
    contradicting the fact that $\mathcal{B}^{\omega'}_n$ is a binary $B_2$-sequence.

Next, we generalize the derivation of an information-theoretic argument from~\cite{della2022note}. Uniformly 
at random pick a pair from $\mathcal{D}$ and denote the outcome by a pair of random variables $(\boldsymbol{X},\boldsymbol{Y})$. Then, the difference $\boldsymbol{X}-\boldsymbol{Y}$ is uniformly distributed over $\mathcal{Z}\backslash\{0^f\}$ (i.e, conditioned on $\boldsymbol{X}-\boldsymbol{Y}\ne 0^f$). Let $E$ be the event $\{\boldsymbol{X}-\boldsymbol{Y}\ne 0^f\}$. Then,
$H(\boldsymbol{X},\boldsymbol{Y})$ equals
    \begin{align}\label{eq:entropyeq}
        &H(\boldsymbol{X},\boldsymbol{Y},\boldsymbol{X}-\boldsymbol{Y})=H(\boldsymbol{X}-\boldsymbol{Y}) 
        +H(\boldsymbol{X},\boldsymbol{Y}|\boldsymbol{X}-\boldsymbol{Y}) \nonumber\\
        &=H(\boldsymbol{X}-\boldsymbol{Y})+\text{Pr}(E)H(\boldsymbol{X},Y|\boldsymbol{X}-\boldsymbol{Y},E)\nonumber\\
        &+\text{Pr}(E^c)H(\boldsymbol{X},\boldsymbol{Y}|\boldsymbol{X}-\boldsymbol{Y},E^c).
    \end{align}
Clearly, $H(\boldsymbol{X},\boldsymbol{Y}|\boldsymbol{X}-\boldsymbol{Y},E^c)=0$, since different 
nonzero elements in $\mathcal{Z}$ have multiplicity one. In addition, 
    \begin{align}
        H(\boldsymbol{X},\boldsymbol{Y}|\boldsymbol{X}-\boldsymbol{Y},E)=\log B^{\omega'}_n, \;\;
        \text{Pr}(E)=\frac{B^{\omega'}_n}{|\mathcal{D}|}\le \frac{r}{B^{\omega'}_n}, \notag
    \end{align}
    where the inequality follows from~\eqref{eq:cauchy}.
    Therefore,
    \begin{align}\label{eq:smallterms}
        &\text{Pr}(E)H(\boldsymbol{X},\boldsymbol{Y}|\boldsymbol{X}-\boldsymbol{Y},E)
        \le \frac{r}{B^{\omega'}_n}\log B^{\omega'}_n.
    \end{align}
    Combining~\eqref{eq:entropyeq},~\eqref{eq:smallterms} with $H(\boldsymbol{X},\boldsymbol{Y})=\log |\mathcal{D}|\ge \frac{(B^{\omega'}_n)^2}{r}$, we obtain
    \begin{align}\label{eq:bbound}
        \log (\frac{(B^{\omega'}_n)^2}{r})\le H(X-\boldsymbol{Y})+\frac{r}{B^{\omega'}_n}\log B^{\omega'}_n.
    \end{align}
In order to obtain an upper bound on $B^{\omega'}_n$, we need an upper bound on $H(\boldsymbol{X}-\boldsymbol{Y})$ specialized for constant-weight vectors.

    Let $n_{ij}$, $i\in[f]$, $j\in[r],$ be the number of suffixes in $\mathcal{S}_j$ whose $i$th coordinate is $1$.   
    By subadditivity of the entropy function we have 
    \begin{align}\label{eq:hxminusy}
        H(\boldsymbol{X}-\boldsymbol{Y})\le \sum^f_{i=1}H(X_i-Y_i) \le fH(\frac{\sum^f_{i=1}\boldp^i}{f}),
    \end{align}
    where $\boldp^i=(p^i_0,p^i_{1},p^i_{-1})$ is the distribution of $X_i-Y_i$,
     \begin{align}\label{eq:pi0}
         &p^i_0=\frac{\sum^r_{j=1}[n^2_{ij}+(|\mathcal{S}_j|-n_{ij})^2]}{\sum^r_{j=1}|\mathcal{S}_j|^2}, \;\;
         p^i_1=\frac{\sum^r_{j=1}n_{ij}(|\mathcal{S}_j|-n_{ij})}{\sum^r_{j=1}|\mathcal{S}_j|^2},\nonumber\\
         &p^i_{-1}=\frac{\sum^r_{j=1}n_{ij}(|\mathcal{S}_j|-n_{ij})}{\sum^r_{j=1}|\mathcal{S}_j|^2}.
     \end{align}
We show next that the average distribution, denoted as
     \begin{align}
     \boldp^*= \frac{\sum^f_{i=1}\boldp^i}{f},  
     \end{align}
satisfies $p^*_0\ge \frac{(\omega'')^2+(f-\omega'')^2}{f^2}$, where $\omega''=\omega-\omega'$ is the weight for all suffixes of codevectors in $\mathcal{B}^{\omega'}_n$, i.e., the weight of vectors in $\mathcal{S}_j$, $j\in[r]$. From \eqref{eq:pi0}, it follows 
     \begin{align}\label{eq:pstar0}
        p^*_0=  &\frac{\sum^f_{i=1}p^i_0}{f}=\frac{\sum^r_{j=1}\big(\sum^f_{i=1}[n^2_{ij}+(|\mathcal{S}_j|-n_{ij})^2]\big)}{f(\sum^r_{j=1}|\mathcal{S}_j|^2)}\nonumber\\
        \ge&\frac{\sum^r_{j=1}\big(\frac{(\sum^f_{i=1}n_{ij})^2}{f}+\frac{(\sum^f_{i=1}(|\mathcal{S}_j|-n_{ij})^2}{f}\big)}{f(\sum^r_{j=1}|\mathcal{S}_j|^2)}\nonumber\\
        \overset{(a)}{=}&\frac{\sum^r_{j=1}\big(|\mathcal{S}_j|^2(\omega'')^2+|\mathcal{S}_j|^2(f-\omega'')^2\big)}{f^2(\sum^r_{j=1}|\mathcal{S}_j|^2)}\nonumber\\
        =&\frac{(\omega'')^2+(f-\omega'')^2}{f^2},
     \end{align}
    where $(a)$ follows from the fact that the weight of the vectors in $\mathcal{S}_j$ is fixed and equal to $\omega''$. In addition, we have  $p^*_1=p^*_{-1}$ since $p^i_{1}=p^i_{-1}$ from \eqref{eq:pi0}. Note that $p^*_0\ge \frac{1}{2}$ and that the entropy function $H(p^*_0,\frac{1-p^*_0}{2},\frac{1-p^*_0}{2})$ is decreasing in $p^*_{0}$ when $p^*_{0}\ge \frac{1}{2}$. Therefore, combined with \eqref{eq:pstar0} and \eqref{eq:hxminusy}, we have  
    \begin{align}\label{eq:hxminusy1}
       H(\boldsymbol{X}-\boldsymbol{Y})\le fH(p_1,p_0,p_{-1}),
    \end{align}
    where $p_0=\frac{(\omega'')^2+(f-\omega'')^2}{f^2}$ and $p_1=p_{-1}=\frac{1-p_0}{2}$.
    Combining  \eqref{eq:bbound} \eqref{eq:hxminusy1}, 
    we obtain 
    \begin{align}\label{eq:logb}
       \log B^{\omega'}_n\le \frac{1}{2}(\log r +fH(p_0,p_1,p_{-1}))+ \frac{r}{2B^{\omega'}_n}\log B^{\omega'}_n.
    \end{align}
    Finally, to prove \eqref{eq:weightdalai1}, suppose to the contrary that
    \begin{align}\label{eq:logb2}
 &\log B^{\omega'}_n> e\,H\left(\frac{\omega'}{e}\right)+\log n, \text{ and}\nonumber\\
  &\log B^{\omega'}_n>\frac{1}{2}\Big[e\,H\left(\frac{\omega'}{e}\right)+fH(p_0,p_1,p_{-1})\Big]+1, 
\end{align}
From \eqref{eq:logb}, \eqref{eq:logb2}, and the fact that $r\le 2^{eH(\frac{\omega'}{e})}$, we have
\begin{align}\label{eq:contradiction}
    1<\frac{r}{2B^{\omega'}_n}\log B^{\omega'}_n,
\end{align}
as well as the inequality below which contradicts~\eqref{eq:contradiction}:
\begin{align}
    B^{\omega'}_n>n2^{e\,H\left(\frac{\omega'}{e}\right)}\ge nr\ge r\log B^{w'}_n.
\end{align}
\vspace{-0.1in}
 \end{proof}
By combining Lemma \ref{lemma:weightdalai} and \eqref{eq:aomega} we conclude that
\begin{align}
    \log A^{\omega}_n\le& \max_{\omega'\in[\max\{0,\omega-f\},\min\{e,\omega\}]}\max\big\{e\,H\left(\frac{\omega'}{e}\right), \nonumber\\
  &\frac{1}{2}\Big[e\,H\left(\frac{\omega'}{e}\right)+fH(p_0,p_1,p_{-1})\Big]\big\}+\log n,
\end{align}
where $p_0=\frac{(\omega'')^2+(f-\omega'')^2}{f^2}$,  $p_1=p_{-1}=\frac{1-p_0}{2}$, and 
$\omega''=\omega-\omega'$, for any choice of $e$ and $f$ such that $e+f=n$. Therefore, 
an upper bound on $R_{\bar\omega}$ is given by 
\begin{align}\label{eq:upperbounddalai}
    R_{\bar\omega}\le&\max_{\bar e\in[0,1]} \max_{\bar\omega'\in[\max\{0,\bar\omega-1+\bar e\},\min\{\bar e,\bar\omega\}]}\max\big\{\bar e\,H\left(\frac{\bar\omega'}{\bar e}\right), \nonumber\\
  &\frac{1}{2}\Big[\bar e\,H\left(\frac{\bar\omega'}{\bar e}\right)+(1-\bar e)H(p_0,p_1,p_{-1})\Big]\big\},
\end{align}
where $p_0=\frac{(\bar\omega'')^2+(1-\bar e-\bar\omega'')^2}{(1-\bar e)^2}$,  $p_1=p_{-1}=\frac{1-p_0}{2}$, and 
$\bar\omega''=\bar\omega-\bar\omega'$. Note that the bound \eqref{eq:upperbounddalai} is smaller than the bound $0.6$ in~\cite{lindstrom1972b2} whenever $\omega< \frac{n}{2}$, and is smaller than the best known upper bound $0.5753$ for unconstrained binary $B_2$-sequences reported in~\cite{cohen2001binary} whenever $\omega\le 0.345n$. See Table \ref{tab:results} for more details regarding the actual values of the upper bounds. 

\section{A Lower Bound} \label{sec:lowerbound}
We describe next a construction for constant-weight binary $B_2$ codes $\mathcal{A}^\omega_n$ of size $(\frac{n}{\omega})^{\frac{\omega}{2}+o(\omega)}$ and 
$$(\left\lfloor\frac{n}{\omega}\right\rfloor)^{\frac{\omega\lceil\frac{n}{\omega}\rceil-n}{2(\lceil\frac{n}{\omega}\rceil-\lfloor\frac{n}{\omega}\rfloor)}}(\left\lceil\frac{n}{\omega}\right\rceil)^{\frac{n-\omega\lfloor\frac{n}{\omega}\rfloor}{2(\lceil\frac{n}{\omega}\rceil-\lfloor\frac{n}{\omega}\rfloor)}}2^{o(\omega)},$$ 
for the case that $\frac{n}{\omega}$ is an integer and a noninteger real value, respectively. The construction implies that $R_{\bar\omega}\ge \frac{\bar\omega}{2}\log (\frac{1}{\bar\omega})$ whenever $\frac{1}{\bar\omega}$ is an integer, and $$R_{\bar\omega}\ge \frac{\bar\omega\lceil\frac{1}{\bar\omega}\rceil-1}{2(\lceil\frac{1}{\bar\omega}\rceil-\lfloor\frac{1}{\bar\omega}\rfloor)}\log (
\left\lfloor\frac{1}{\bar\omega}\right\rfloor)+\frac{1-\bar\omega\lfloor\frac{1}{\bar\omega}\rfloor}{2(\lceil\frac{1}{\bar\omega}\rceil-\lfloor\frac{1}{\bar\omega}\rfloor)}\log (\left\lceil\frac{1}{\bar\omega}\right\rceil)$$ otherwise. An important observation is that our construction, although conceptually simple, results in codes with rate at least $\frac{1}{4}$th of the largest possible rate of unconstrained constant-weight codes, $\binom{n}{\omega}$. 

In what follows, we assume for simplicity that $\frac{n}{\omega}$ is an integer, and as before, we let $\omega\le \frac{n}{2}$. The idea is to find a surjective linear mapping $F:\{0,1,2\}^n\rightarrow [0,(\frac{n}{\omega})^\omega-1]$ that converts any length-$n$ vector over the alphabet $\{0,1,2\}$ into an integer in $[0,(\frac{n}{\omega})^\omega-1]$. More precisely, the mapping $F$ is required to satisfy the following two properties:
\begin{enumerate}
\item[(A)] For any integer $i\in[0,(\frac{n}{\omega})^\omega-1]$, there exists a vector $\boldc\in\{0,1\}^n$ of weight $\omega$, such that $F(\boldc)=i$.
    \item[(B)] For any vectors $\boldc,\boldc'\in\{0,1\}^n$, we have that $F(\boldc)+F(\boldc')=F(\boldc+\boldc')$. Note that $\boldc+\boldc'\in\{0,1,2\}^n$.
\end{enumerate} 
Given the mapping $F$, we construct an integer Sidon set~\cite{ruzsa1998infinite} from the set $[0,(\frac{n}{\omega})^\omega-1]$. By the Bose-Chawla construction, there exists a set of integers $\{i_1,\ldots,i_{(\frac{n}{\omega})^{\frac{\omega}{2}+o(\omega)}}\}\subset [0,(\frac{n}{\omega})^\omega-1]$ of size $(\frac{n}{\omega})^{\frac{\omega}{2}+o(\omega)}$ such that the sums of any two integers in the set are distinct. Then from property (A) of the mapping $F$, for every $i_j$, $j\in [1,(\frac{n}{\omega})^{\frac{\omega}{2}+o(\omega)}]$, there exists a vector $\boldc_j\in\{0,1\}^n$ of weight $\omega$  such that $F(\boldc_j)=i_j$. Finally, by property (B) of the mapping $F$ and the definition of the set $\{i_1,\ldots,i_{(\frac{n}{\omega})^{\frac{\omega}{2}+o(\omega)}}\}$, the set $\{\boldc_1,\ldots,\boldc_{(\frac{n}{\omega})^{\frac{\omega}{2}+o(\omega)}}\}$ is a binary $B_2$ codebook of weight $\omega$.

For any integer $k\in[0,n-1]$, let $k=a_k\omega+b_k$, where $a_k=\lfloor \frac{k}{\omega}\rfloor$ and $b_k=k\bmod \omega$. For any $\boldc\in\{0,1,2\}^n$ define
 \begin{align}
     F(\boldc)\triangleq \sum^n_{i=1}a_{i-1}\left(\frac{n}{\omega}\right)^{b_{i-1}}c_i.
 \end{align}
 It is obvious that $F$ satisfies property (B). To show that $F$ satisfies (A), we note that any integer $m\in [0,(\frac{n}{\omega})^\omega-1]$ has a $\frac{n}{\omega}$-ary representation $m = \sum^{\omega-1}_{i=0}m_i\,\left(\frac{n}{\omega}\right)^i,$
where $m_i\in [0,\frac{n}{\omega}-1]$. Let $\boldc_m$ be a vector in $\{0,1\}^n$, whose indices of the $1$ bits are given by $\{m_i\omega+i:i\in[0,\omega-1]\}$. Then, $\boldc_m$ has weight $\omega$ and $F(\boldc_m)=\sum^{\omega-1}_{i=0}a_{m_i\omega+i}(\frac{n}{\omega})^{b_{m_i\omega+i}}= m$. Hence $F$ satisfies property (B).	


\bibliographystyle{IEEEtran}
\bibliography{biblio1,biblio2}
\newpage
\appendix
\section{Proof of Lemma \ref{lemma:weightlindstrom}}\label{sec:proofoflindstrom}
\vspace{-0.1in}

We split every vector $\boldc=\bolda\boldb\in \mathcal{A}^\omega_n$ into $\bolda\in\{0,1\}^e$ and $\boldb\in\{0,1\}^f$. Next, we group the vectors $\boldc\in\mathcal{A}^\omega_n$ based on the weight of $\bolda$, and use the definitions $\mathcal{B}^{\omega'}=\{\boldc:\boldc=\bolda\boldb\in\mathcal{A}^\omega_n,\;\bolda\in\{0,1\}^e,\;|\bolda|=\omega'\}$. 
\begin{lemma}\label{lemma:weightlindstrom}
For any $\omega'\in[\max\{0,\omega-f\},\min\{e,\omega\}]$, we have
\begin{align}\label{eq:weight}
  \log |\mathcal{B}^{\omega'}_n|\le &\max\Bigg\{H(\frac{\omega'}{e})\cdot e+\log (n+1)+1, \nonumber\\
  &\frac{1}{2}\Big[H(\frac{\omega'}{e})\cdot e+\log (n+1)+H(\frac{2\omega''}{f})\cdot f\nonumber\\
  &+2\omega''+\log (n(n+1))\Big]+1\Bigg\}, 
\end{align}
where $\omega''=\min\{\frac{f}{4},\omega-\omega'\}$ and $H(x)=-x\log_2x-(1-x)\log_2(1-x)$ is the entropy function.
\end{lemma}
\begin{proof}
The first part of the proof follows along the same lines as that of Lemma~\ref{lemma:weightdalai}. We set $\{\bolda_1,\ldots,\bolda_r\}=\{\bolda:\bolda\in\{0,1\}^e,\bolda\boldb\in\mathcal{B}^{\omega'}_n \text{ for some $\boldb$}\}$, $\mathcal{S}_i=\{\boldb:\bolda_i\boldb\in\mathcal{B}^{\omega'}_n\}$, and $\mathcal{D}=\cup^r_{i=1}\{(\boldb_1,\boldb_2):\boldb_1,\boldb_2\in\mathcal{S}_i\}$. This establishes~\eqref{eq:cauchy}. The remainder of the proof differs from the one provided for Lemma~\ref{lemma:weightdalai}, as we use combinatorial arguments~\cite{lindstrom1972b2}. 

 Let $\{\boldv_1,\ldots,\boldv_{|\mathcal{D}|}\}=\{\boldb_1-\boldb_2:(\boldb_1,\boldb_2)\in \mathcal{D}\}$ be the set of vectors that are the differences of the pairs of suffixes in $\mathcal{D}$, with multiplicities. Then, the multiplicity of $0^f$ in $\{\boldv_1,\ldots,\boldv_{|\mathcal{D}|}\}$ is $B^{\omega'}_n$ and the multiplicity of each nonzero vector in $\{\boldv_1,\ldots,\boldv_{|\mathcal{D}|}\}$ is one. For $i\in[|\mathcal{D}|],$ $j\in[1,f],$ let $h_{ij}=1$, if the the $j$-th bit of $\boldv_i$ is $0$, and $h_{ij}=-1$ otherwise. Then, $\sum^{|\mathcal{D}|}_{i=1}h_{ij}\ge 0$ for $j\in[1,f]$ and $\sum^{|\mathcal{D}|}_{i=1}\sum^f_{j=1}h_{ij}\ge 0$.
    
Since $\sum^f_{j=1}h_{ij}=f-2k$ for each $\boldv_i$ with $k$ non-zero entries, we have $\sum^f_{j=1}h_{ij}\le 0$ for any $\boldv_i$ having at least $\frac{f}{2}$ non-zero entries. In addition, the number of possible difference vectors with $k$ non-zero entries is at most $\binom{f}{k}2^k$. Let $|\boldv_i|$ denote the number of nonzero entries in $\boldv_i$, $i\in[|\mathcal{D}|]$. Then,
    \begin{align}\label{eq:sumofhx}       
    0\le \sum^{|\mathcal{D}|}_{i=1}\sum^f_{j=1}h_{ij} &\overset{(a)}{\le} fB^{\omega'}_n+\sum_{i:|\boldv_i|\ge 1}(\sum^f_{j=1}h_{ij})\nonumber\\
        &\overset{(b)}{\le}fB^{\omega'}_n+\sum_{i:1\le|\boldv_i|\le \min\{\frac{f}{2},2(\omega-\omega')\}}(\sum^f_{j=1}h_{ij})\nonumber\\
        -\sum_{i:|\boldv_i|> \frac{f}{2}}1 &\overset{(c)}{\le} fB^{\omega'}_n+ 2f\omega''\binom{f}{2\omega''}2^{2\omega''} -\sum_{i:|\boldv_i|>\frac{f}{2}}1, \nonumber\\
    \end{align}
where $\omega'' = \min\{\omega-\omega',\frac{f}{4}\}$, $(a)$ follows from the fact that the multiplicity of $0^f$ in $\{\boldv_1,\ldots,\boldv_s\}$ is $B^{\omega'}_n$, $(b)$ follows from the fact that $\sum^f_{j=1}h_{ij}< 0$ when $|\boldv_i|> \frac{f}{2}$ and the fact that $|\boldv_i|\le 2(\omega-\omega')$, and $(c)$ follows from the fact that the number of $\boldv_i$ with $|\boldv_i|=k$ is at most $\binom{f}{k}2^k$ and the fact that this number increases with $k$ when $k\le \min\{\frac{f}{2},2(\omega-\omega')\}$. 
    Eq. \eqref{eq:sumofhx} implies that
    \begin{align}\label{eq:sumofxupper}
        |\mathcal{D}|\le &(f+1)|\mathcal{B}^{\omega'}_n|+2(f+1)\omega''\binom{f}{2\omega''}2^{2\omega''}\nonumber\\
        \le&(f+1)|\mathcal{B}^{\omega'}_n|+2^{fH(\frac{2\omega''}{f})+2\omega''+\log (n(n+1))}
    \end{align}
Combining \eqref{eq:cauchy} and \eqref{eq:sumofxupper}, we have 
    \begin{align*}
         |\mathcal{B}^{\omega'}_n|^2\le (f+1)|\mathcal{B}^{\omega'}_n|r+2^{fH(\frac{2\omega''}{f})+2\omega''+\log (n(n+1))}r.
    \end{align*}
    Since $r\le 2^{eH(\frac{\omega'}{e})}$, we then have~\eqref{eq:weight}. Based on Lemma \ref{lemma:weightlindstrom} and \eqref{eq:aomega}, we have \eqref{eq:upperboundlindstrom}.
\end{proof}
\end{document}